\documentclass[11]{article}

\usepackage{amsfonts,latexsym,graphicx,amssymb,amsmath,enumitem,physics}
\usepackage{amsthm}
\usepackage{fullpage}
\usepackage{color} 
\usepackage{hyperref,cleveref}

\mathchardef\mhyphen="2D
\usepackage{mathtools}
\DeclarePairedDelimiter{\ceil}{\lceil}{\rceil}
\newcommand{\bit}{\{0,1\}}
\newcommand{\seclam}{1^\lambda}

\newcommand{\negl}{\mathsf{negl}}
\newcommand{\poly}{\mathsf{poly}}

\newcommand{\N}{\mathbb{N}}
\newcommand{\Z}{\mathbb{Z}}

\newcommand{\Q}{\mathbb{Q}}

\newcommand{\Rq}{R_q}
\newcommand{\Zq}{\Z_q}
\newcommand{\st}{^*}

\newcommand{\set}[1]{\left\{#1\right\}}

\newcommand{\brackets}[1]{\ensuremath{\left( #1 \right)}}
\newcommand{\Brackets}[1]{\ensuremath{\left[ #1 \right]}}

\newcommand{\A}{\ensuremath{\mathcal{A}}}
\newcommand{\B}{\ensuremath{\mathcal{B}}}

\newcommand{\D}{\ensuremath{\mathcal{D}}}

\newcommand{\F}{\ensuremath{\mathcal{F}}}

\newcommand{\K}{\ensuremath{\mathcal{K}}}

\newcommand{\calO}{\ensuremath{\mathcal{O}}}
\renewcommand{\P}{\ensuremath{\mathcal{P}}}
\newcommand{\R}{\ensuremath{\mathcal{R}}}

\newcommand{\X}{\ensuremath{\mathcal{X}}}
\newcommand{\Y}{\ensuremath{\mathcal{Y}}}

\newtheorem{theorem}{Theorem}[section]
\newtheorem{claim}{Claim}[theorem]

\newtheorem{definition}{Definition}[theorem]

\newtheorem{lemma}{Lemma}[theorem]

\newtheorem{assumption}{Assumption}

\newlength{\protowidth}
\newcommand{\pprotocol}[5]{
{\begin{figure*}[#4]
\begin{center}
\setlength{\protowidth}{\textwidth}

        {
        \hrulefill \vspace{5pt}
        \small
        {\quad
        \begin{minipage}{\protowidth}
        \begin{center}
        {\bf #1}
        \end{center}
        #5

        \hrulefill

        \end{minipage}
        \quad}
        }

        \caption{\label{#3} #2}
\end{center}
\vspace{-4ex}
\end{figure*}
} }

\newcommand{\protocol}[4]{
\pprotocol{#1}{#2}{#3}{tbh!}{#4} }

\newcommand{\rlweqchi}{\mathsf{RLWE}_{R, q,\chi}}
\newcommand{\Rn}{R_n}
\newcommand{\Rnqn}{R_{n, q_n}}
\newcommand{\Rnq}{R_{n,q}}

\newcommand{\gentrap}{\textsc{GenTrap}}
\newcommand{\inv}{\textsc{Invert}}

\newcommand{\trap}{\tau}
\newcommand{\trapf}{\kappa}
\newcommand{\supp}{\textsc{Supp}}
\newcommand{\genf}{\textsc{Gen}_{\mathcal{F}}}
\newcommand{\invf}{\textsc{Inv}_{\mathcal{F}}}
\newcommand{\chkf}{\textsc{Chk}_{\mathcal{F}}}
\newcommand{\sampf}{\textsc{Samp}_{\mathcal{F}}}
\newcommand{\J}{\mathsf{BitDecomp}}
\newcommand{\bdx}{\overline{x}}

\newcommand{\Es}[1]{E_{#1}}
\renewcommand{\vec}[1]{\mathbf{#1}}
\newcommand{\mata}{\vec{a}}
\newcommand{\mate}{\vec{e}}

\newcommand{\cnt}{\mathsf{count}}

\newif\ifnotes\notestrue

\ifnotes
\usepackage{color}
\definecolor{mygrey}{gray}{0.50}
\newcommand{\notename}[2]{{\textcolor{mygrey}{\footnotesize{\bf (#1:} {#2}{\bf ) }}}}

\newcommand{\pnote}[1]{{\endnote{#1}}}

\else

\newcommand{\notename}[2]{{}}

\newcommand{\pnote}[1]{}

\fi

\begin{document}

\title{Simpler Proofs of Quantumness}

\author{Zvika Brakerski\\Weizmann Institute of Science\\\texttt{zvika.brakerski@weizmann.ac.il}\thanks{Supported by the Binational Science Foundation (Grant No. 2016726), and by the European Union Horizon 2020 Research and Innovation Program via ERC Project REACT (Grant 756482) and via Project PROMETHEUS (Grant 780701).} 
\and Venkata Koppula\\Weizmann Institute of Science\\\texttt{venkata.koppula@weizmann.ac.il}\thanks{Supported by the Binational Science Foundation (Grant No. 2016726), and by the European Union Horizon 2020 Research and Innovation Program via ERC Project REACT (Grant 756482) and via Project PROMETHEUS (Grant 780701).} 
\and Umesh Vazirani\\University of California Berkeley\\\texttt{vazirani@cs.berkeley.edu} \thanks{Supported in part by ARO Grant W911NF-12-1-0541, NSF Grant CCF1410022, a Vannevar Bush faculty fellowship, and the Miller Institute at U.C. Berkeley through a
Miller Professorship.} 
\and Thomas Vidick\\California Institute of Technology\\\texttt{ vidick@caltech.edu} \thanks{Supported  by  NSF CAREER Grant CCF-1553477, AFOSR YIP award number FA9550-16-1-0495, a CIFAR Azrieli Global Scholar award,
  MURI Grant FA9550-18-1-0161, and the IQIM, an NSF Physics Frontiers
Center (NSF Grant PHY-1125565).}}

\date{}

\maketitle

\begin{abstract}
	A proof of quantumness is a method for provably demonstrating (to a classical verifier) that a quantum device can perform computational tasks that a classical device with comparable resources cannot. Providing a proof of quantumness is the first step towards constructing a useful quantum computer. 
	
	There are currently three approaches for exhibiting proofs of quantumness: $(i)$ Inverting a classically-hard one-way function (e.g.\ using Shor's algorithm). This seems technologically out of reach. $(ii)$ Sampling from a classically-hard-to-sample distribution (e.g.\ BosonSampling). This may be within reach of near-term experiments, but for all such tasks known verification requires exponential time. $(iii)$ Interactive protocols based on cryptographic assumptions. The use of a trapdoor scheme allows for efficient verification, and implementation seems to require much less resources than $(i)$, yet still more than $(ii)$. 
	
	In this work we propose a significant simplification to approach $(iii)$ by employing the random oracle heuristic. (We note that we \emph{do not} apply the Fiat-Shamir paradigm.)
	
	We give a two-message (challenge-response) proof of quantumness based on any trapdoor claw-free function. In contrast to earlier proposals we do not need an adaptive hard-core bit property. This allows the use of smaller security parameters and more diverse computational assumptions (such as Ring Learning with Errors), significantly reducing the quantum computational effort required for a successful demonstration. 
	
\end{abstract}

\section{Introduction}

Quantum computing holds a promise of a qualitative leap in our ability to perform important computational tasks. These tasks include simulation of chemical and physical systems at the quantum level, generating true randomness, algorithmic tasks such as factoring large numbers, and more. However, constructing a quantum computer with capabilities beyond those of existing classical computers is technologically challenging. Indeed, whether it is possible or not remains to be proven; such a ``proof'' is the focus of the ongoing race to construct a useful quantum device, with records for device size and functionality set at an increasing rate by the likes of Google, IBM, and the increasing number of startups heavily invested in this race. 
This notion, known as ``proof of quantumness'',\footnote{The term ``quantum supremacy'' is also used in the literature.} is generally viewed as a major milestone towards unlocking the powers of quantum computing. 
We can classify existing approaches towards proof of quantumness into three families:
\begin{enumerate}
    \item There are tasks that are generally believed to be classically intractable, and for which quantum algorithms are known; most notably the factoring and discrete logarithm problems \cite{Shor94}. Constructing a quantum computer that can factor beyond our classical capabilities would constitute a valid proof of quantumness. Alas, in order to implement the factoring algorithm on relevant input sizes one requires fault-tolerant quantum computation, which seems technologically out of reach (see e.g.~\cite{gidney2019factor} for recent and highly optimized estimates ranging in the millions of qubits).
    
    \item A different approach, introduced independently by Bremner, Josza and Shepperd~\cite{bremner2010classical} and by Aaronson and Arkhipov \cite{AA11}, is to use a quantum device to sample from distributions that are presumed to be hard to sample from classically. The intractability of classically achieving the task has not stood the same test of time as more established problems such as e.g.\ factoring, but can nonetheless be based on reasonable complexity-theoretic conjectures, at least for the problem of exact sampling. While quantum devices that can sample from these distributions appear to be ``right around the corner'', the real challenges are in (i) showing hardness of approximate sampling --- the quantum device will never be perfect --- and (ii) the classical verification: verification for these methods generally requires investing exponential classical computational resources, and can thus only be performed for fairly small input lengths.
    
    \item A new approach was recently proposed in \cite{BCMVV18}. They propose to use \emph{post-quantum cryptography}, namely to rely on cryptographic assumptions that cannot be broken even by the quantum device. Rather than verifying that the quantum device has the ability to break the assumption, cryptography is used to compel the device to generate a quantum superposition in a way that can be efficiently verified using a secret key. This method is inherently interactive, unlike the previous two, and requires at least four rounds of communication. As a cryptographic building block it uses trapdoor claw-free function families (recall that claw-freeness was originally introduced in the context of digital signatures and constructed based on factoring \cite{GMR84sigs}). In addition to claw freeness, the \cite{BCMVV18} approach also requires an additional adaptive hardcore bit property which appears to be hard to realize and is currently only known to be achievable based on the Learning with Errors (LWE) assumption \cite{Regev05}.
\end{enumerate}

The third approach is compelling in its ability to verify quantumness even of large quantum devices efficiently, but it still requires a large number of quantum operations. Furthermore, the interactive nature of the protocol requires the quantum device to retain a superposition while waiting for the verifier's second message (a single random bit).

In this work we simplify the \cite{BCMVV18} approach and allow for it to be based on a more diverse set of computational assumptions. This marks a step towards a protocol that can be realistically implemented on an actual quantum device, and can be efficiently verified on a classical computer. 

\paragraph*{Our Results.} We propose to use the \emph{random oracle heuristic} as a tool to reduce the round complexity of the proof of quantumness protocol from~\cite{BCMVV18}, making it into a simple one-round message-response protocol. We note that it is unlikely that a similar result can be achieved \emph{in the standard model} without introducing an additional hardness assumption. The reason is that a single-round message-response protocol in the standard model (i.e.\ without oracles) immediately implies that quantum samplers cannot be efficiently de-quantized (otherwise the protocol will have no soundness). Such a result therefore implies a (weak) separation between the BQP and BPP models. However, the LWE assumption does not appear to imply such a separation, and the current state of the art suggests that it is equally intractable in the quantum and classical settings.\footnote{This insight is due to a discussion with Omer Paneth.}

We show that using the random oracle heuristic, it is possible to implement the protocol in a single round while at the same time eliminating the need for an adaptive hard-core bit property, and thus relying on any family of claw-free functions. In particular, we propose a construction of trapdoor claw free functions which is analogous to that of \cite{BCMVV18} but relies on the Ring-LWE assumption \cite{LPR10,LPR13}. Ring-LWE based primitives are often regarded as more efficient than their LWE-based counterparts since they involve arithmetic over polynomial rings, which can be done more efficiently than over arbitrary linear spaces. Despite the similarity between LWE and Ring-LWE, proving an adaptive hard-core theorem for the latter appears to be a challenging task. This is since the LWE-based construction uses a so-called lossiness argument that is not known to be replicable in the Ring-LWE setting. We note that we can also instantiate our method using ``pre-quantum'' cryptography since soundness should hold only with respect to classical adversaries. Using a back-of-the-envelope calculation we estimate that it is possible to execute our protocol using superpositions over $\sim 8 \lambda\log^2\lambda$ qubits, for security parameter $\lambda$ and the adversary would have advantage negligible in $\lambda$.

While we allow the use of trapdoor claw-free families based on arbitrary assumptions, which should  allow for better security/efficiency trade-offs, our protocol still requires the quantum device to evaluate the random oracle on a quantum superposition, which could potentially create an additional burden. We point out that current and future heuristic instantiations of the random oracle model using explicit hash functions are assumed to enjoy efficient quantum implementation. Specifically, in evaluating the \emph{post-quantum} security level of cryptographic constructions (e.g.\ for the NIST competition \cite{NIST}), security is evaluated in the Quantum Random Oracle model where adversaries are assumed to evaluate hash functions on superpositions as efficiently as they do classically. Granted, this is just a model for an adversary, but it is customary to try to be as realistic as possible and not over-estimate the power of the adversary. We therefore consider the evaluation of the random hash function as a relatively lower-order addition to the cost of performing the quantumness test. 

Lastly, we compare our method to the most straightforward way to employ a random oracle for the purpose of round reduction, the Fiat-Shamir transform \cite{FiatShamir}. The basic protocol of \cite{BCMVV18} contains $4$ messages, where the third message is simply a random bit. One can therefore do parallel repetition of the protocol (though the soundness of this transformation needs to be shown),\footnote{Very recently, two concurrent works by Alagic et al.~\cite{ACGH19} and Chia et al.~\cite{CCY19} showed that parallel repetition of Mahadev's protocol indeed achieves negligible soundness error.} and apply Fiat-Shamir to compress it into challenge-response form. Furthermore, for proofs of quantumness soundness is only required to hold against a classical adversary, so the standard security reduction for Fiat-Shamir should hold. This approach only requires to apply the random oracle to a classical input. However, it still requires the adaptive hard-core bit property and is therefore restricted to the LWE assumption. We believe that our protocol, being of a somewhat modified form compared to prior works, may be useful for future applications.

\paragraph*{Our Technique.} At a high level, a family of trapdoor claw free functions allows to sample a function $f: \{0,1\} \times \{0,1\}^n \to \{0,1\}^n$ together with a trapdoor.  The function has two branches $f(0,\cdot), f(1, \cdot)$ which are both injective, i.e.\ permutations (this is a simplified description, actual protocols use a relaxed ``noisy'' notion). It is guaranteed that it is computationally intractable to find a collision (``claw'') $x_0, x_1$ s.t.\ $f(0,x_0)=f(1,x_1)$, however given the trapdoor it is possible to find for all $y$ the preimages $x_0, x_1$ s.t.\ $f(0,x_0)=f(1,x_1)=y$.

The \cite{BCMVV18} protocol sends a description of $f$ to the quantum device, asks it to apply $f$ on a uniform superposition of inputs and measure the image register, call the value obtained $y$. The quantum device is then left with a uniform superposition over the two preimages of $y$: $(0,x_0)$ and $(1,x_1)$. The value $y$ is sent to the verifier who challenges the quantum device to measure the remaining superposition on inputs in either the standard or Hadamard basis. A classical adversary that can answer each query independently must also be able to answer both at the same time, which is ruled out by the adaptive hard core property.

We propose to enable the quantum device to generate a superposition over $(0,x_0, H(0,x_0))$ and $(1,x_1, H(1,x_1))$, where $H$ is a one-bit hash function modeled as a random oracle. This can be done in a straightforward manner, similar to the previous method. The device is then asked to measure the resulting state in the Hadamard basis (always), and send the outcomes obtained to the verifier.\footnote{In fact we use a slight variant of this protocol, since measuring the $H$ part in Hadamard basis has probability $1/2$ of erasing the information on that bit. Instead we append the $H$ values directly to the phase. This is immaterial for the purpose of this exposition.} Since the device makes a single measurement, there is no need for a challenge from the verifier,  which effectively collapses the protocol to two messages. A quick calculation shows that the verifier receives a bit $m$ and vector $d$ s.t.\ in the case of a honest behavior the equation $m = d\cdot(x_0 \oplus x_1) \oplus H(0,x_0) \oplus H(1,x_1)$ holds. Finally, the verifier uses the trapdoor to recover $x_0, x_1$ from $y$ and checks that the equation is satisfied. The crux of the security proof is that a classical adversary cannot query the oracle at both $(0,x_0)$ and $(1,x_1)$, otherwise it would have been able to find a claw and break the cryptographic assumption. Therefore at least one value out of $H(0,x_0)$ and $H(1,x_1)$ remains random, and thus the adversary cannot compute $m, d$ that adhere to the required equation with probability greater than $1/2$. The proof thus follows from a simple extraction-style argument. In our main protocol, we use parallel repetition to argue that no prover can succeed with non-negligible probability. 

\paragraph*{Discussion on the `Random Oracle' Heuristic} As discussed above, the Fiat-Shamir heuristic can be used for the quantum supremacy protocol of Brakerski et al.~\cite{BCMVV18}. However, this would mean that the resulting scheme would require stronger assumptions (in particular, it would require noisy TCFs with the adaptive hardcore bit property). Secondly, starting with the work of Canetti et al.~\cite{CGH04}, many works have shown uninstantiability of the random oracle. These works show certain cryptographic primitives which are secure in the random oracle model, but are broken when instantiated by any concrete hash function. However, these constructions are very contrived, and in particular, do not apply to our protocol. 

\paragraph*{Efficiency of our Protocol, and Comparison to Previous Approaches}

We would like to emphasize that at the current level of maturity of quantum technology, any estimate of `practical advantage' would be educated guesswork at best. The technology for any option is far from being available and it is hard to predict the direction that technology will take, and as a consequence the practical cost of implementing certain operations.

This state of affairs, we believe, highlights the importance of developing multiple approaches to tasks such as proof of quantumness. This way, an assortment of solutions will be ready to accommodate the different directions that technology may lead.

A second point that we wish to highlight before getting into technical calculations, is that our approach allows to use \emph{any} family of trapdoor claw free permutations (and as we point out, for proofs of quantumness even `pre-quantum' candidates will suffice, e.g. if a candidate can be devised based on DDH in EC groups). This means that our back of the envelope calculation only refers to one specific way of using our scheme. Currently, we do not know any candidates for trapdoor claw free permutations based on such `pre-quantum' assumptions. 

Our protocol can be executed using a quasi-linear number of qubits and, with the proper choice of candidate for the hash function, has quasi-linear computational complexity. 

\emph{Comparison with \cite{BCMVV18}}: Since we do not require the hardcore bit property, our input dimension $n$ is smaller by a factor of at least $60 \log(\lambda)$. This follows due to Lemma 4.2 in \cite{BCMVV18}. Also, note that the parameter $q$ must also grow, hence the overall number of qubits required to implement the protocol in \cite{BCMVV18} is $O(\lambda \log^3(\lambda))$, at least $100 \log(\lambda)$ times more. Secondly, since \cite{BCMVV18} is a four-round protocol, the prover must maintain its quantum state until it receives a challenge from the verifier.

\emph{Comparison to discrete log via Shor's algorithm}: Let $n$ denote the number of bits required for representing the group elements. The current estimates for the number of qubits required for discrete log are ~$3n$, while the number of quantum gates required is ~$0.3 n^3$ (see \cite{gidney2019factor}). Similar to Shor's algorithm for factoring/discrete log, our protocol is also a non-interactive one (that is, the verifier sends a challenge, and the prover responds with an answer).

\paragraph*{Open Problems.} Our work suggests a number of open problems in the context of utilizing random oracles in the regime of classical verification of quantum computation. Most desirably, whether it is possible to use the random oracle in order to eliminate the need for other assumptions, or at least the need for a trapdoor. Obtaining a publicly verifiable protocol is a highly desirable goal. We can also wonder whether our protocol can be used for the purposes of certifying randomness or verifying quantum computation. In the plain model, the adaptation of the proof of quantumness for these purposes was far from trivial and yet the protocol itself is almost identical. Improving the state of the art in certifying randomness and in verifiability using random oracles (or using other methods) is also an interesting open problem.


\section{Preliminaries}
\label{sec:prelims}

\subsection{Notations}

For an integer $n$ we write $[n]$ for the set $\{1,\ldots,n\}$.
For any finite set $X$, let $x\gets \X$ denote a uniformly random element drawn from $X$. Similarly, for any distribution $\D$, let $x \gets \D$ denote a sample from $\D$. For an element $x\in X$ we write $\J(x)$ for an arbitrarily chosen but canonical (depending only on the implicit set $X$) binary representation of $x$. For any density function $f$ on domain $X$, let $\supp(f)$ denote the support of $f$; that is $\supp(f) = \set{x \in X : f(x) > 0}$.

For density functions $f_1, f_2$ over the same finite domain $X$, the Hellinger distance between $f_1$ and $f_2$ is $$ H^2(f_1, f_2) = 1 - \sum_{x \in X} \sqrt{f_1(x) f_2(x)}.$$

The total variation distance between $f_1$ and $f_2$ is $$\norm{f_1 - f_2}_{\mathrm{TV}} = \frac{1}{2} \sum_{x \in X} |f_1(x) - f_2(x)| \leq \sqrt{2H^2(f_1, f_2)}.$$

The following lemma relates the Hellinger distance and the trace distance of superpositions.

\begin{lemma}
    Let $X$ be a finite set and $f_1, f_2$ two density functions on $X$. Let $$ \ket{\psi_1} = \sum_{x \in X} \sqrt{f_1(x)} \ket{x}, \text{ and } \ket{\psi_2} = \sum_{x \in X} \sqrt{f_2(x)} \ket{x}.$$  
    Then
    $$\norm{\ket{\psi_1} - \ket{\psi_2}}_{\mathrm{tr}} \leq \sqrt{1 - (1 - H^2(f_1, f_2))^2}. $$
\end{lemma}




\subsection{Ideal Lattices}

In this section, we present some background on ideal lattices, the truncated discrete Gaussian distribution and the Ring Learning with Errors problem. For a positive integer $B$, modulus $q$, and dimension $n$, the truncated discrete Gaussian distribution is a distribution with support $\set{x \in \Zq^n : \norm{x} \leq B\sqrt{n}}$ defined as follows: 
$$ D_{\Zq^n, B}(x) = \frac{\exp(-\pi \norm{x}^2/B^2)}{\sum_{z \in \Zq^n, \norm{z} \leq B\sqrt{n}} \exp(-\pi |z|^2/B^2)}.$$ 

The Ring Learning with Errors (RLWE) assumption\cite{LPR13} is parameterized by a ring $R$, modulus $q \in \N$ and a noise distribution $\chi$. Informally, the assumption states that given many samples of the form $(a, a\cdot s + e)$ where $s$ is fixed for all samples, $a$ is chosen uniformly at random and $e$ is chosen from the error distribution $\chi$ for each sample, it is hard to compute $s$. The formal definition is given below. Here, we restrict ourselves to a special family of \emph{cyclotomic rings}.

\begin{assumption}
	Let $n$ be a power of two, $f_n(X) = X^n + 1$ an irreducible polynomial over $\Q[X]$ and $\Rn = \Z[X]/(f_n(X))$. Let $q = \set{q_n}_{n \in \N}$ be a family of moduli, $\Rnqn = \Rn/q_n \Rn = \Z_{q_n}[X]/(f_n(X))$ the quotient space, and $\chi = \set{\chi_n}_{n \in \N}$ a family of error distributions, where $\chi_n$ is a distribution over $\Rnqn$. For any secret $s$ in $\Rnqn$, let $\calO_s$ denote the oracle that, on each query, chooses $a \gets \Rnqn$, $e \gets \chi_n$ and outputs $(a, a\cdot s + e \mod q_n)$. The Ring Learning with Errors assumption $\rlweqchi$, parameterized by the family of rings $\set{\Rn}_{n = 2^k, k \in \N}$, moduli family $q$ and distribution family $\chi$, states that for any PPT adversary $\A$, there exists a negligible function $\negl(\cdot)$ such that for all security parameters $n = 2^k, k \in \N$, $$ \Pr[s \gets \A^{O_s()}(1^n) : s \gets \Rnqn ] \leq \negl(n).$$ 
\end{assumption}

Given many samples $\set{a_i, a_i \cdot s + e_i}_{i}$, one can efficiently find $s$ using a \emph{trapdoor} for the public elements $\set{a_i}_i$. There exists a sampling algorithm that can sample $\set{a_i}_{i}$ together with a trapdoor $\trap$, and an inversion algorithm that uses $\trap$ to extract $s$ from the set of evaluations $\set{a_i \cdot s + e_i}_{i}$. Without the trapdoor, the public elements $\set{a_i}_i$ look uniformly random. 

\begin{theorem}[Theorem 5.1 of \cite{MP12} in the Ring setting]
	\label{thm:trapdoor}
	Let $n, m, q$ be such that $n$ is a power of $2$,  $m = \Omega(\log q)$. There is an efficient randomized algorithm $ \gentrap$ that takes as input $(1^n, 1^m, q)$, and returns $ \mata = \brackets{a_i}_{i} \in \Rnq^m$ and a trapdoor $\trap$ such that the distribution of $\mata$ is negligibly (in $n$) close to the uniform distribution over $\Rnq^m$. Moreover, there is an efficient algorithm $\inv$ and a universal constant $C_T$ such that the following holds with overwhelming probability over the choice of $(\mata, \trap) \gets \gentrap(1^n, 1^m, q)$: 
	$$ \text{for all } s \in \Rnq, \mate \text{ such that }\norm{\mate} \leq \frac{q}{C_T\sqrt{n \log q}},  \inv(\mata, \trap, \mata\cdot s + \mate) = s. $$ 
\end{theorem}

\subsection{Noisy Trapdoor Claw-Free Hash Functions}

In this section we introduce the notion of noisy trapdoor claw-free functions (NTCFs). Let $\X, \Y$ be finite sets and $\K$ a set of keys. For each $k \in K$ there should exist two (efficiently computable) injective functions $f_{k,0}, f_{k,1}$ that map $\X$ to $\Y$, together with a trapdoor $t_k$ that allows efficient inversion from $(b,y)\in \bit \times \Y$ to $f_{k,b}^{-1}(y)\in\X \cup \set{\perp}$. For security, we require that for a randomly chosen key $k$, no polynomial time adversary can efficiently compute $x_0, x_1 \in \X$ such that $f_{k,0}(x_0) = f_{k,1}(x_1)$ (such a pair $(x_0, x_1)$ is called a \emph{claw}). 

Unfortunately, we do not know how to construct such `clean' trapdoor claw-free functions. Hence, as in previous works \cite{BCMVV18,Mah18}, we will use `noisy' version of the above notion. For each $k \in \K$, there exist two functions $f_{k,0}, f_{k,1}$ that map $\X$ to a distribution over $\Y$.

The following definition is taken directly from \cite{BCMVV18}.

\begin{definition}[NTCF family]\label{def:trapdoorclawfree}
	Let $\lambda$ be a security parameter. Let $\X$ and $\Y$ be finite sets.
	 Let $\mathcal{K}_{\mathcal{F}}$ be a finite set of keys. A family of functions 
	$$\mathcal{F} \,=\, \big\{f_{k,b} : \X\rightarrow \mathcal{D}_{\Y} \big\}_{k\in \mathcal{K}_{\mathcal{F}},b\in\{0,1\}}$$
	is called a \emph{noisy trapdoor claw free (NTCF) family} if the following conditions hold:

	\begin{enumerate}
		\item{\textbf{Efficient Function Generation.}} There exists an efficient probabilistic algorithm $\genf$ which generates a key $k\in \mathcal{K}_{\mathcal{F}}$ together with a trapdoor $t_k$: 
		$$(k,t_k) \leftarrow \genf(1^\lambda)\;.$$

		\item{\textbf{Trapdoor Injective Pair.}}  
		\begin{enumerate}
			\item \textit{Trapdoor}: There exists an efficient deterministic algorithm $\invf$ such that with overwhelming probability over the choice of $(k,t_k) \gets \genf(\seclam)$, the following holds: \\
			$$ \text{for all } b\in \{0,1\}, x\in \X \text{ and } y\in \supp(f_{k,b}(x)), \invf(t_k,b,y) = x. $$
			\item \textit{Injective pair}: For all keys $k\in \mathcal{K}_{\mathcal{F}}$, there exists a perfect matching $\R_k \subseteq \X \times \X$ such that $f_{k,0}(x_0) = f_{k,1}(x_1)$ if and only if $(x_0,x_1)\in \R_k$. 
		\end{enumerate}

		\item{\textbf{Efficient Range Superposition.}}
		For all keys $k\in \mathcal{K}_{\mathcal{F}}$ and $b\in \{0,1\}$ there exists a function $f'_{k,b}:\X\to \mathcal{D}_{\Y}$ such that the following hold.
		\begin{enumerate} 
			\item For all $(x_0,x_1)\in \mathcal{R}_k$ and $y\in \supp(f'_{k,b}(x_b))$, $\invf(t_k,b,y) = x_b$ and $\invf(t_k,b\oplus 1,y) = x_{b\oplus 1}$. 
			\item There exists an efficient deterministic procedure $\chkf$ that, on input $k$, $b\in \{0,1\}$, $x\in \X$ and $y\in \Y$, returns $1$ if  $y\in \supp(f'_{k,b}(x))$ and $0$ otherwise. Note that $\chkf$ is not provided the trapdoor $t_k$. 
			\item For every $k$ and $b\in\{0,1\}$, 
			$$ \Es{x\leftarrow_U \X} \big[\,H^2(f_{k,b}(x),\,f'_{k,b}(x))\,\big] \,\leq\, 1/50\;.\footnote{Here $1/50$ can be replaced by an arbitrarily small constant.}$$
			 Here $H^2$ is the Hellinger distance. Moreover, there exists an efficient procedure  $\sampf$ that on input $k$ and $b\in\{0,1\}$ prepares the state
			\begin{equation}
			    \frac{1}{\sqrt{|\X|}}\sum_{x\in \X,y\in \Y}\sqrt{(f'_{k,b}(x))(y)}\ket{x}\ket{y}\;.
			\end{equation}
		\end{enumerate}

		\item{\textbf{Claw-Free Property.}}
		For any PPT adversary $\A$, there exists a negligible function $\negl(\cdot)$ such that the following holds: 
		\begin{align*}
			\Pr\Brackets{(x_0, x_1) \in \R_k  : (k, t_k) \gets \genf(\seclam), (x_0, x_1) \gets \A(k)} \leq \negl(\lambda)
		\end{align*}

	\end{enumerate}

\end{definition} 

\section{Proof of Quantumness Protocol}
We will now present our protocol. Throughout the protocol, we will ignore dependence on the security parameter when clear from context. Let $\F$ be a NTCF family with domain $\X$, range $\Y$ described by the algorithms $\genf, \invf, \chkf, \sampf$. Let $w$ denote the length of bit decomposition of elements of $\X$. Finally, let $H$ be a hash function that maps $\X$ to $\bit$.

\protocol{Proof of Quantumness Protocol}{Protocol for Proof of Quantumness}{fig:supremacy}{
    The protocol is parameterized by a hash function $H : \bit^n \to \bit$ (which will be modeled as a random oracle in the security proof).
    
	\begin{enumerate}
		\item The verifier generates $(k, {\trapf}) \gets  \genf(\seclam)$ and sends $k$ to the prover. 
		\item The prover sends $\lambda$ tuples $\set{(y_i, m_i, d_i)}_{i \in [\lambda]}$. The verifier initializes $\cnt = 0$ and performs the following checks: 
		\begin{enumerate}
			\item It checks that all values in $\set{y_i}_{i}$ are distinct. 
			\item It computes $x_{i,b} = \invf({\trapf}, b, y_i)$ for each $i \in [\lambda]$, $b \in \bit$. Next, it  checks if $m_i = d_i^T \cdot (\J(x_{i,0}) + \J(x_{i,1})) + H(x_{i,0}) + H(x_{i,1})$. If this check passes, it increments the value of $\cnt$ by $1$. 
		\end{enumerate}
		\item If $\cnt > 0.75 \lambda$, the verifier outputs $1$, else it outputs $\perp$. 
	\end{enumerate}
}

\begin{theorem}
	Let $\F$ be a family of NTCF functions satisfying Definition \ref{def:trapdoorclawfree}. Then Protocol \ref{fig:supremacy} satisfies the following properties:
	\begin{itemize}
		\item[-] Completeness: There exists a quantum polynomial-time prover $\P$ and a negligible function $\negl(\cdot)$ such that for all $\lambda \in \N$ and hash functions $H$, $\P$ succeeds in the protocol with probability at least $1 - \negl(\lambda)$. 
		\item[-] Proof of Quantumness: For any PPT (classical) adversary $\A$, there exists a negligible function $\negl(\cdot)$ such that for all $\lambda \in \N$, $\A$ succeeds in the protocol with probability at most $\negl(\lambda)$ where $H$ is modeled as a random oracle.  
 	\end{itemize}
 	
\end{theorem}

\subsection{Completeness} 
\label{sec:protocol-completeness}
In this section, we show that the honest (quantum) prover is accepted by the verifier. 

\newcounter{honest-state}[section]
The honest prover receives  NTCF key $k$. It does the following:
\begin{enumerate}
	\item It starts with $\lambda$ copies of the state $\ket{0}\ket{0}\ket{0}\ket{-}$. For each $i \in [\lambda]$, let $\ket{\psi_i} = \ket{0}\ket{0}\ket{0}\ket{-}$. It then applies $\sampf$ to the first three registers of $\ket{\psi_i}$ for each $i$, resulting in the state $\ket{\psi'^{(1)}_i}$, where 
	\begin{equation} 
		\label{real-term}
		\ket{\psi_1'^{(\stepcounter{honest-state}\arabic{honest-state})}}  = \brackets{\frac{1}{\sqrt{2|\X|}}\sum_{x \in \X, y \in \Y, b \in \bit} \sqrt{(f_{k,b}'(x))(y)} \ket{b} \ket{x} \ket{y} } \ket{-}. 
	\end{equation}

	This quantum state is at distance at most $0.2$ from the following quantum state: 
	\begin{equation} 
		\label{ideal-term}
		\ket{\psi_i^{(\arabic{honest-state})}}  = \brackets{\frac{1}{\sqrt{2|\X|}}\sum_{x \in \X, y \in \Y, b \in \bit} \sqrt{(f_{k,b}(x))(y)} \ket{b} \ket{x} \ket{y} } \ket{-}. 
	\end{equation}

	\item Next, it measures the third register, obtaining measurement $y \in \Y$. Let $x_{0}, x_{1} \in \X$ be the unique elements such that $y$ is in the support of $f_{k,b}(x_{b})$. Applying this operation to the state in (\ref{ideal-term}), the resulting state (ignoring the measured register) is 
	\begin{equation}
		\label{ideal-term-2}
		\ket{\psi_i^{(\stepcounter{honest-state}\arabic{honest-state})}} = \brackets{ \frac{1}{\sqrt{2}} \brackets{\ket{0} \ket{x_{0}} + \ket{1}\ket{x_{1}}}} \ket{-}.
	\end{equation} 
	
	\item Let $U_{H}$ be a unitary that maps $\ket{a}\ket{b}$ to $\ket{a}\ket{b+H(a)}$. The prover applies $U_H$ to the second and third register. On applying this operation to the state in \eqref{ideal-term-2}, the new state is 
	\begin{equation}
		\label{ideal-term-3}
		\ket{\psi_i^{(\stepcounter{honest-state}\arabic{honest-state})}} =  \frac{1}{2} \brackets{ \sum_{b,b'} (-1)^{b'} \ket{b}\ket{x_{b}}\ket{b'+H(x_{b})} }.
	\end{equation}
	
	\item The prover then evaluates the function $\J$ on the second register. Applying this to \eqref{ideal-term-3}, the resulting state is 
	\begin{equation}
		\label{ideal-term-4}
		\ket{\psi_i^{(\stepcounter{honest-state}\arabic{honest-state})}} =  \frac{1}{2} \brackets{ \sum_{b,b'} (-1)^{b'} \ket{b}\ket{\J(x_{b})}\ket{b'+H(x_{b})} }. 
	\end{equation}
	
	\item Finally, the prover applies the Hadamard operator to all registers. On applying this to \eqref{ideal-term-4}, this produces the state  (where $h_{b} = H(x_{b})$ and $\bdx_{b} = \J(x_{b})$)
	\begin{align}
		\ket{\psi_i^{(\stepcounter{honest-state}\arabic{honest-state})}} &= \frac{1}{\sqrt{2^{w+4}}} \sum_{b,b' \in \bit} \sum_{\substack{m, m' \in \bit,\\ d \in \bit^w}} (-1)^{m\cdot b + d^T\cdot \bdx_{b} + m'\cdot b' + m'\cdot h_{b} + b'} \ket{m}\ket{d}\ket{m'}\notag \\
		&= 	\frac{1}{\sqrt{2^{w+2}}} \sum_{m\in \bit, d\in \bit^w} \ket{m}\ket{d}\ket{1} \brackets{(-1)^{d^T\cdot \bdx_{0} + h_{0}} + (-1)^{m + d^T\cdot x_{1} + h_{1}}}	\label{ideal-term-5}
	\end{align}
	Upon measurement of the state in (\ref{ideal-term-5}), the output tuple $(m,d,1)$ satisfies $m  = d^T\cdot (\bdx_{0} + \bdx_{1}) + h_{0} + h_{1}$ (with probability $1$). As a result, applying the above operations to $\ket{\psi'^{(1)}_i}$ results in a tuple $(y,m,d)$ that is accepted with probability at least $0.8$. Using a Chernoff bound it is straightforward to argue that there exists a negligible function $\negl(\cdot)$ such that with probability at least $1-\negl(\lambda)$, at least $3/4$ fraction of the tuples in $\set{(y_i, m_i, d_i)}$ pass the verification. 
	\end{enumerate}

\subsection{Proof of Quantumness : Classical Prover's Advantage in the Random Oracle Model}

Here, we will show that if the function $H$ is replaced with a random oracle, then any classical algorithm that has non-negligble advantage in Protocol \ref{fig:supremacy} can be used to break the claw-free property of $\F$. Consider the following security experiment which captures the interaction between a (classical) prover and a challenger in the random oracle model; the challenger represents the verifier in the protocol. 

\paragraph{Experiment 1} In this experiment, the challenger represents the verifier in Protocol \ref{fig:supremacy} and also responds to the random oracle queries issued by the prover.
\begin{enumerate}
	\item The challenger (verifier) chooses an NTCF key $(k,\trapf) \gets \genf(\seclam)$ and sends $k$ to the prover. The prover and challenger have access to a random oracle $H : \bit^n \to \bit$. 

	\item The prover sends $\set{(y_i,m_i,d_i)}_{i \in [\lambda]}$. For each $i \in [\lambda]$, the challenger computes $x_{i,b} \gets \invf(\trapf, b, y_i)$ for $b \in \bit$, queries the random oracle $H$ on $x_{i,0}, x_{i,1}$ and receives $h_{i,0}, h_{i,1}$ respectively. Next, it checks if $m_i = d_i^T\cdot (\J(x_{i,0}) + \J(x_{i,1})) + h_{i,0} + h_{i,1}$. If at least $0.75 \lambda$ tuples satisfy the check, it outputs $1$, else it outputs $\perp$. 
\end{enumerate}

\paragraph{Experiment 2} This experiment is similar to the previous one, except that the challenger implements the random oracle, and does not use the trapdoor for performing the final $\lambda$ checks. 

\begin{enumerate}

	\item The challenger (verifier) chooses an NTCF key $(k,\trapf) \gets \genf(\seclam)$ and sends $k$ to the prover. The challenger also implements the random oracle as follows. It maintains a database which is initially empty. On receiving a query $x$, it checks if there exists a tuple $(x, h)$ in the database. If so, it outputs $h$, else it chooses a random bit $h \gets \bit$, adds $(x,h) $ to the database and outputs $h$.

	\item The prover sends $\set{(y_i,m_i,d_i)}_{i \in [\lambda]}$. On receiving this set from the prover, the challenger does not compute the inverses of $y_i$. Instead, it initializes $\cnt = 0$, and for each $i$, it looks for tuples $(x_{i,0}, h_{i,0})$ and $(x_{i,1}, h_{i,1})$ in the table such that $\chkf(y_i, 0, x_{i,0}) = \chkf(y_i, 1, x_{i,1}) = 1$. If such $(x_{i,0}, x_{i,1})$ do not exist, then the challenger chooses a random bit $r_i$ and sets $\cnt = \cnt + r_i$. Else, it checks if $m_i = d_i^T \cdot (\J(x_{i,0}) + \J(x_{i,1})) + h_{i,0} + h_{i,1}$. If so, it increments $\cnt$. 

	Finally, it checks if $\cnt > 0.75 \lambda$. If so, it outputs $1$, else outputs $\perp$.
\end{enumerate}

\paragraph{Experiment 3} This experiment is identical to the previous one, except that the challenger, after receiving $\set{(y_i, m_i, d_i)}_i$, outputs $\perp$ if for all $i \in [\lambda]$, there does not exist two entries $(x_{i,0}, h_{i,0}), (x_{i,1}, h_{i,1})$ in the database such that $\chkf(y_i, 0, x_{i,0}) = \chkf(y_i, 1, x_{i,1}) = 1$. 

\paragraph{Analysis}: For any classical PPT prover $\A$, let $p_{\A}$ denote the probability that the verifier outputs $1$ in Protocol \ref{fig:supremacy} (when $H$ is replaced with a random oracle), and for $w \in \set{1,2, 3}$, let $p_{\A,w}$ denote the probability that the challenger interacting with $\A$ in Experiment $w$ outputs $1$. From the definition of Experiment 1 it follows that $p_{\A} = p_{\A,1}$. 

\begin{claim}
For any prover $\A$, $p_{\A,1} = p_{\A,2}$. 
\end{claim}

\begin{proof}
The main differences between Experiment 1 and Experiment 2 are that the challenger implements the random oracle, and secondly, after receiving $\set{(y_i, m_i, d_i)}_i$, the challenger does not use the trapdoor for checking. Note that in Experiment 1, if either $x_{i,0}$ or $x_{i,1}$ are not queried to the random oracle $H$, then $H(x_{i,0}) + H(x_{i,1})$ is a uniformly random bit. Moreover, since the $y_i$ values are distinct, if there exist two indices $i,j$ such that both the preimages of $y_i$ and $y_j$ are not queried, then $H(x_{i,0}) + H(x_{i,1})$ is independent of $H(x_{j,0}) + H(x_{j,1})$. As a result, for each index $i$ such that the preimages of $y_i$ are not queried, the value of $\cnt$ is incremented with probability $1/2$. 

In Experiment 2, the challenger checks for pairs corresponding to $x_{i,0}$ and $x_{i,1}$ in the database, and if either of them is missing, it increments $\cnt$ with probability $1/2$. As a result, the probability of $\cnt > 0.75 \lambda$ is identical in both experiments. 
\end{proof}

\begin{claim}
There exists a negligible function $\negl(\cdot)$ such that for any prover $\A$ and any security parameter $\lambda \in \N$, $p_{\A,2} \leq p_{\A,3} + \negl(\lambda)$.
\end{claim}

\begin{proof}
The only difference between these two experiments is that the challenger, at the end of the experiment, outputs $\perp$ if for all $i \in [\lambda]$, either $x_{i,0}$ or $x_{i,1}$ has not been queried to the random oracle. The only case in which the challenger outputs $1$ in Experiment 2 but outputs $\perp$ in Experiment 3 is when for all $i \in [\lambda]$, either $x_{i,0}$ or $x_{i,1}$ has not been queried, but there exist $t \geq 0.75 \lambda$ indices $\set{i_1, \ldots, i_t}$ such that $\cnt$ was incremented. Using Chernoff bounds, we can show that this happens with negligible probability.
\end{proof}

\begin{claim}
 Assuming $\F$ is a secure claw-free trapdoor family, for any PPT prover $\A$, there exists a negligible function $\negl(\cdot)$ such that for all $\lambda \in \N$, $p_{\A,3}(\lambda) \leq \negl(\lambda)$. 
\end{claim}

\begin{proof}
Suppose there exists a PPT prover $\A$ and a non-negligible function $\epsilon(\cdot)$ such that for all $\lambda \in \N$, the challenger outputs $1$ with probability $\epsilon = \epsilon(\lambda)$ in Experiment 3. This means with probability at least $\epsilon$, there exists an index $i\st \in [\lambda]$ such that $\A$ queries the random oracle on $x_{i\st,0}, x_{i\st,1}$ and finally outputs $\set{(y_i, m_i, d_i)}_{i}$ such that $\chkf(y_{i\st}, 0, x_{i\st, 0}) = \chkf(y_{i\st}, 1, x_{i\st,1}) = 1$. 

We will construct a reduction algorithm $\B$ that breaks the claw-free property of $\F$ with probability $\epsilon$. The reduction algorithm receives the key $k$ from the NTCF challenger, which it forwards to $\A$. Next, $\A$ makes polynomially many random oracle queries, which are answered by the reduction algorithm by maintaining a database. Eventually, $\A$ sends $\set{(y_i,m_i,d_i)}$. The reduction algorithm checks if there exist tuples $(x_{i\st, 0}, h_{i\st, 0})$ and $(x_{i\st,1}, h_{i\st,1})$ in its database such that $\chkf(y_{i\st}, 0, x_{i\st,0}) = \chkf(y_{i\st}, 1, x_{i\st,1}) = 1$. If so, it sends $(x_{i\st,0}, x_{i\st,1})$ to the NTCF challenger. 
\end{proof}

Using the above claims, it follows for every classical prover $\A$, there exists a negligible function $\negl(\cdot)$ such that for all $\lambda \in \N$, $p_{\A} \leq \negl(\lambda)$. 

\section{Construction of NTCFs based on Ring LWE}

Our construction is similar to the one in \cite{BCMVV18}. Let $\lambda$ be the security parameter, $n = 2^{\ceil{\log \lambda}}$. The following are other parameters chosen by our scheme (we will ignore dependence on security parameter/$n$): 
\begin{itemize}
	\item Ring $R = \Z[X]/(X^n + 1)$.  
	\item Modulus $q = \poly(n)$, $\Rq = R/qR$
	\item $m = \Omega(\log q)$ : determines the dimension of range space
	\item $\chi$: the noise distribution. In our case, $\chi$ is a Discrete Gaussian over $\Z^{n}$ with parameter $B_V$.
	\item $B_P$ : the noise bound for function evaluation. We require the following constraints on $B_P$: 
	\begin{itemize}
		\item $B_P \geq \Omega(n\cdot m \cdot B_V)$
		\item $2 B_P \sqrt{n\cdot m} \leq q/(C_T \cdot \sqrt{n \log q}) $ for some constant $C_T$
	\end{itemize}
\end{itemize}

The domain is $\X = \Rq$, and range is $\Y = \Rq^m$. 

Each function key $k = (\mata, \mata\cdot s + \mate)$, where $s \in \Rq$, $a_i, e_i \in \Rq$ for all $i \in [m]$, $\mata = [a_1 \ldots a_m]^T$, $\mate = [e_1 \ldots e_m]^T$. For $b \in \bit, x \in \X, k = (\mata, \mata\cdot s + \mate)$, the density function $f_{k,b}(x)$ is defined as follows: 
\begin{align}
	\forall \vec{y} \in \Y, (f_{k,b}(x))(\vec{y}) = D_{\Z^{n\cdot m}, B_P}(\vec{y} - \mata \cdot x - b \cdot \mata \cdot s),
\end{align} 
where $\vec{y} = [y_1 \ldots y_m]^T$, and each $\vec{y}_i$ can be represented as an element in $\Zq^n$ (using the coefficient representation); similarly for $\mata\cdot x$ and $\mata\cdot s$.

We will now show that each of the properties of NTCFs hold. 

\begin{enumerate}
	\item \textbf{Efficient Key Generation}: The key generation algorithm $\genf(\seclam)$ first chooses $(\mata, \trap) \gets \gentrap(1^n, 1^m, q)$, $s \gets \Rq$ and $\mate \gets \chi^m$. It outputs key $k = (\mata, \mata \cdot s + \mate)$, and the trapdoor is $\trapf = (\trap, k, s)$.  

	\item \textbf{Trapdoor Injective Pair}: 
	\begin{enumerate}
		\item \emph{Trapdoor}: For $k = (\mata, \mata \cdot s + \mate)$, $x \in \X$ and $b\in \bit$, the support of $f_{k,b}(x)$ is  $$\supp(f_{k,b}(x)) = \set{ \vec{y} \in \Y : \norm{\vec{y} - \mata\cdot x - b \cdot \mata\cdot s} \leq B_P\sqrt{n\cdot m}}$$ 

		The inversion algorithm $\invf$ takes as input the lattice trapdoor $\trap$, $b \in \bit$, $\vec{y} \in \Y$ and outputs $\inv(\trap, \mata, \vec{y}) - b\cdot s$. From Theorem \ref{thm:trapdoor}, it follows that with overwhelming probability over the choice of $\mata$, for all $\vec{y} \in \supp(f_{k,b}(x))$, $\inv(\trap, \mata, \vec{y}) = x + b\cdot s$. Hence, it follows that $\invf(\trapf, b, \vec{y}) = x$.

		\item \emph{Injective Pair}: Let $k = (\mata, \mata\cdot s + \mate)$. From the construction, it follows that $f_{k,0}(x_0) = f_{k,1}(x_1)$ if and only if $x_1 = x_0 + s$. Hence the set $\R_k = \set{(x, x+s) : x \in \X}$. 

	\end{enumerate}

	\item \textbf{Efficient Range Superposition}: The function $f'_{k,0}$ is same as $f_{k,0}$, while $f'_{k,1}$ is defined as follows (recall $k = (\mata, \mata \cdot s + \mate)$): 
	\begin{align}
		\forall \vec{y} \in \Y, (f'_{k,1}(x))(\vec{y}) = D_{\Z^{n\cdot m}, B_P}(\vec{y} - \mata \cdot x - (\mata \cdot s + \mate))
	\end{align} 
	\begin{enumerate}
		\item Since $f'_{k,0} = f_{k,0}$, it follows that for all $(x_0, x_1) \in \R_k$ and $\vec{y} \in \supp(f'_{k,0}(x_0))$, $\invf(\trapf, 0, \vec{y}) = x_0$ and $\invf(\trapf, 1, \vec{y}) = x_1$. We need to show the same for $f'_{k,1}$; that is, for all $(x_0, x_1) \in \R_k$ and $\vec{y} \in \supp(f'_{k,1}(x_1))$, $\invf(\trapf, 1, \vec{y}) = x_1$ and $\invf(\trapf, 0, \vec{y}) = x_0$. For all $x \in \X$, $$ \supp(f'_{k,1}(x)) = \set{\vec{y} \in \Y : \norm{\vec{y} - \mata\cdot x - \mata\cdot s - \mate} \leq B_P\sqrt{n\cdot m}}$$

		Hence for any $\vec{y} \in \supp(f'_{k,1}(x))$,  $\norm{\vec{y} - \mata\cdot x_1 - \mata \cdot s} \leq  2B_P \sqrt{n\cdot m}$; using Theorem \ref{thm:trapdoor}, we can conclude that $\invf(\trapf, 1, \vec{y}) = x_1$. 

		\item The procedure $\chkf$ takes as input $\vec{y} \in \Y, k = (\mata, \vec{v}), b \in \bit, x \in \X$ and checks if $\norm{\vec{y} - \mata\cdot x - b\cdot \vec{v}} \leq B_P \sqrt{n\cdot m}$. 

		\item The definition of $\sampf$ is identical to the one in \cite{BCMVV18}, and the Hellinger distance can be bounded by $1-e^{-\frac{2\pi m \cdot n \cdot B_V}{B_P}}$.
		From our setting of parameters, this quantity is at most $1/50$.

	\end{enumerate}
	\item \textbf{Claw-Free Property} Suppose there exists an adversary $\A$ that, on input $k = (\mata, \mata \cdot s + \mate)$ can output $(x_0, x_1) \in \R_k$. Then this adversary can be used to break the Ring LWE assumption, since $x_1 - x_0 = s$.
\end{enumerate}	 

\bibliographystyle{alpha}

\bibliography{refs}

\end{document}